\documentclass[11pt]{article}
\usepackage{times}
\usepackage{amsmath,amsthm,amssymb,setspace,enumitem,epsfig,titlesec,verbatim,color,array,eurosym,multirow,blkarray,lscape,bm,nicefrac,comment,soul}
\usepackage{lineno}
\usepackage[bf,small]{caption}
\usepackage[top=2.5cm,left=2.8cm,right=2.8cm,bottom=3.4cm]{geometry}
\usepackage{ulem}

\newtheorem{theorem}{Theorem}

\newtheorem{proposition}[theorem]{Proposition}

\newcommand{\keywords}[1]{\textbf{Keywords:}\quad #1}

\allowdisplaybreaks

\title{{\bf The role of behavioural plasticity \\ in finite vs infinite populations}}
\author{
M. Kleshnina$^{1*}$,
K. Kaveh$^{2}$,
K. Chatterjee$^{1}$\\
\\
$^1$IST Austria,
$^2$Dartmouth College\\
\vspace{4mm}
$^*$ Corresponding author, email: maria.kleshnina@ist.ac.at (MK)
}

\begin{document}

\maketitle

\vspace{3mm}

\begin{abstract}
Evolutionary game theory has proven to be an elegant framework providing many fruitful insights in population dynamics and human behaviour. Here, we focus on the aspect of behavioural plasticity and its effect on the evolution of populations. We consider games with only two strategies in both well-mixed infinite and finite populations settings. We assume that individuals might exhibit behavioural plasticity referred to as incompetence of players. We study the effect of such heterogeneity on the outcome of local interactions and, ultimately, on global competition. For instance, a strategy that was dominated before can become desirable from the selection perspective when behavioural plasticity is taken into account. Furthermore, it can ease conditions for a successful fixation in infinite populations' invasions. We demonstrate our findings on the examples of Prisoners' Dilemma and Snowdrift game, where we define conditions under which cooperation can be promoted. 
\end{abstract}



\section*{Introduction}

Evolutionary game theory has been used to study the evolution of populations and their interactions in biology since its first appearance in 1973 \cite{Smith1973}. It proved to be an elegant framework providing a lot of fruitful discoveries in biology and human behaviour. One aspect that evolutionary game theory aims to study is the influence of stochasticity on the evolution of populations. Here, we consider a game setting where individuals might make mistakes when executing their type's strategy by utilising the concept that was first referred to as \emph{incompetence of players} \cite{Beck2013,Filar2012}. We shall refer to this concept as behavioural plasticity that results in a random strategy execution. Such plasticity may influence the outcome of the local interaction affecting the entire population in the long run. 

Behavioural stochasticity is a popular object of study for game-theorists. First, the concept of ``trembling hands'' \cite{Selten1975} was suggested as an approach to players' mistakes during the strategies' execution with some small probability. Later, in evolutionary games it was modelled via mutations \cite{Stadler1992,Tarnita2009}, language learning \cite{Komarova2004,Komarova2001,Nowak2001} or other experimental learning processes \cite{FudenbergLevine,Hopkins2002,McKelvey1995,Selten1991}, adaptation dynamics \cite{Levin2003}, phenotypic plasticity \cite{Dridi2019} and edge diversity in games on graphs \cite{Su2016,Su2019}. Furthermore, the replicator-mutator dynamics \cite{Bomze1995,Nowak2001} was suggested to model mutations in the evolutionary dynamics, where each type has its own mutation rate but these mutations do not occur simultaneously. However, even though stochasticity plays an important role in evolutionary games, especially in finite populations, the possibility of mutants to imitate a resident strategy by mistake has not yet been considered.

Evolutionary games aim to uncover effects of natural selection on the populations and mechanisms that could help their survival \cite{Apaloo2009,Hofbauer2003,Apaloo1995,Nowak2006,Smith1973}. First, infinitely large completely mixed population settings were considered. Replicator dynamics \cite{Taylor1978,Zeeman1980,Hofbauer2003} has been used to predict which strategy might become stable (or evolutionary stable) \cite{Smith1982}. Individuals are chosen at random for interaction and none of them preserves memory about whom they interacted with before. Such settings help us to gain a high-level understanding of the game and strategies comparison which may lead to existence of evolutionary stable strategies (ESS). 

While a well-mixed infinitely large population is an elegant theoretical model, the more realistic model for natural setting is the finite and spatially distributed population  \cite{Durrett1994,Nowak2004a,Allen2017,Kaveh2019}. Here, stochasticity in the interactions and reproductive processes is taken into account affecting the way natural selection works. In such settings coexistence of several strategies does not appear to be a general outcome resulting in one strategy dominating another. That is, if $N$ is the size of the population, then the fixation probability $\rho$ differs from a neutral drift $\nicefrac{1}{N}$. 
Further, in coordination games, fixation of one strategy or another can also depend on the initial frequency of the mutant type. Counter-intuitively, for mutants to be successful their initial relative frequency, $x^*$, should be less than $\nicefrac{1}{3}$, which is referred to as a $\nicefrac{1}{3}$ law \cite{Ohtsuki2007,Nowak2006}. The complexity of the problem grows with the strength of the selection acting on the populations. As a consequence, analytical results are obtained for a weak or no selection. 

In this manuscript, mistakes during the invasion represent a defensive mechanism helping invaders to blend in the resident population. Once established, the invaders may re-learn their own strategy after the competition pressure is neglected. By utilising the notion of plasticity (or incompetence), first we start with the analysis of a general well-mixed finite population prone to behavioural mistakes and analyse the effect of competence on the fixation probabilities. We determine conditions for which one strategy could dominate another one. Second, we expand our result to the infinite population settings. We also allow individuals to improve their strategy execution in order to study the overall effect of incompetence on the outcome of the selection.

The notion of incompetence is related to the concept of strategy plasticity \cite{carja2017evolutionary}. A type or strategy is plastic if it can switch or transform between different subtypes. These subtypes are distinguished by their fitness values. Alternatively, instead of switching, an individuals of a given subtype can asymmetrically give birth to individuals of other types. Such type plasticity creates phenotypic variability in populations. Phenotypic plasticity can be seen as a result of stochastic gene expression without any environmental triggers \cite{raj2008nature} or as a response to environmental challenges \cite{acar2005enhancement}. Bacterial persistence and metabolic switching in microbial population \cite{kussell2005phenotypic, gallie2015bistability,cohen2013microbial} as well as differentiation-dedifferentiation in tumours \cite{kaveh2017stem} are among examples of phenotypic plasticity at a cellular level. 

Evolutionary dynamics for plastic strategies have been mostly discussed in constant selection schemes. One can, however, consider a phenotypically variable population where fitness is derived from a matrix game. Then, the game interaction is defined not only among the competing types but also the subtypes belonging to the same phenotypic group. We can consider Darwinian selection between two plastic strategies where each are composed of different subpopulations. However, the notion of incompetence is distinct from other models of strategy plasticity 
in that only during game interactions a given type can play as a variable type, that is, adopt two or more strategies. However, during reproduction events, it will always create a uniform population. In other words, in the model with incompetent strategies there are no actual distinct subtypes in the populations that coexist at the same time but rather a player can execute two or more strategies during a game-theoretic interaction. Hence, such plasticity, or incompetence, is better applied to behavioural traits of individuals rather than genes and phenotypes.

We demonstrate the influence of behavioural plasticity on 2 examples of social interactions: Prisoners' Dilemma and Snowdrift game \cite{Doebeli2005}. In social dilemmas, an interaction between cooperators and defectors is studied. Cooperation is a ubiquitous and beneficial behavioural trait that is yet prone to exploitation by free-riders. Cooperative populations are often prone to invasions by selfish individuals. However, a population consisting of only free-riders may not survive. In complex ecological systems, elaborate memory-one strategies demonstrate their ability to explain existence of cooperation between ``good'' individuals while punishing selfish defectors \cite{Nowak2006,Nowak2005,Sigmund2010selfish,Sigmund2012,Hilbe2018}. But what if a one-step strategy is not enough in particular settings? We suggest to look into the case when cooperators may take sufficient time to defeat defectors. However, we do not develop a detailed biological model but describe the game and its results in terms of the probability to survive. Our results are applicable to a wide range of scenarios in which payoffs affect survival. 

The manuscript has the following structure. We first introduce the model for a 2x2 matrix game and interpret it. Then, we characterise the model's behaviour in both finite and infinite populations' settings. We derive all possible transitions in a game and a new form of $\nicefrac{1}{3}$ law. Next, we demonstrate the effect of incompetence on the examples of Prisoners' Dilemma and Snowdrift games. In the last section, we discuss the solutions of evolutionary dynamical models of phenotypic plasticity. We derive analytic expression for the steady state solutions in weak selection and for small levels of plasticity. We show that this model does not follow any $\nicefrac{1}{3}$ law similar to that of derived for the model of incompetence.



\section*{Model}

Consider a game between two strategies $1$ and $2$. When an individual playing strategy $1$ interacts with an individual playing strategy $2$, then individual $1$ obtains a reward $b$ and individual $2$ obtains $c$. When two individuals play strategy $1$, they both obtain $a$. When two individuals play strategy $2$, they both obtain $d$. The fitness of individuals is derived from a two-player matrix game. 
However, these game outcomes are possible under the assumption that players can perfectly behave as their own type. That is, behavioural mistakes or uncertainty is not taken into account. Specifically, if interacting individuals can bias their strategy choice by their opponents type, behavioural mistakes can disturb the outcome of the game due to the unforeseen change of the opponent's strategy. An example of such disturbed interaction is depicted in Figure \ref{fig:scheme}. Here, both types with certain probabilities may execute a strategy different from their initial type, disturbing the interaction outcome. Then, the outcome depends not only on the type of the interacting individuals but also on their mistake probabilities. 

\begin{figure}[h!]
\begin{center}
\includegraphics[scale=0.3]{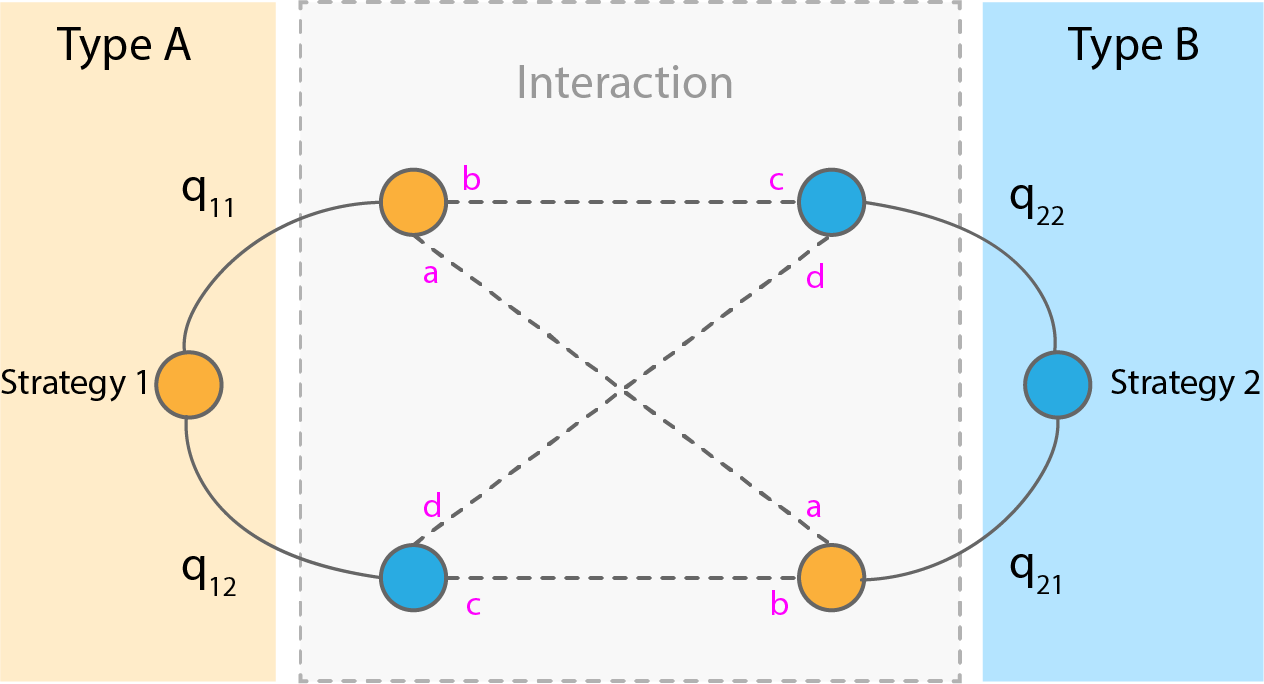}
\end{center}
\caption{A schematic representation of the interaction between type $A$ and $B$ individuals under behavioural uncertainty. Here, both individuals might switch to the strategy of their opponent or execute their own type strategy ($q_{12}=1-\gamma$ and $q_{21}=1-\eta$ versus $q_{11}=\gamma$ and $q_{22}=\eta$). Then, the interaction outcome is stochastic and can be one of the following: interaction between two type $A$ individuals, two type $B$ individuals or type $A$ and $B$ individuals. }
\label{fig:scheme}
\end{figure}

In what follows, we consider two-dimensional games under the assumption that players can play their opponent's strategy with some probabilities. This can be considered as an interaction between players utilising mixed strategies rather than pure types. Individuals are assigned probability measures capturing their mixed strategies in the stochastic incompetence matrix $Q$

\begin{equation*}
\label{Qmatrixdef}
Q=\left(\begin{array}{cc}
q_{11}& q_{12}\\
q_{21}& q_{22}
\end{array} \right).
\end{equation*}

Such settings are formally referred to as 'incompetence' \cite{Filar2012,Kleshnina2017}. The incompetence is introduced in the game by probabilities that each type, A or B, executes its strategy to play the game. A player A executes strategy 1 upon choosing strategy 1 with probability $q_{11}$. With probability $q_{12}$ it might play strategy 2. Similarly, type B individual plays strategy 2 after choosing strategy 2 with probability $q_{22}$ and with probability $q_{21}$ she plays strategy 1.
Since the probabilities in the incompetence matrix, $Q$, add up to unity in each row, we simplify our notation as $q_{11} = \gamma, q_{12}=1-\gamma, q_{21} =1-\eta, q_{22}=\eta$. Thus, the incompetence of mutants is measured with parameter $\gamma$ while $\eta$ denotes the level of incompetence in residents 

\begin{align}
Q = \bordermatrix{~ & {\rm A} & {\rm B} \cr
                  {\rm A} & \displaystyle \gamma & 1-\gamma\cr
                  {\rm B} & 1-\eta & \displaystyle \eta}.
\end{align}

Note that the original payoff matrix is given by

\begin{align}
R = \bordermatrix{~ & {\rm 1} & {\rm 2} \cr
                  {\rm 1} & \displaystyle a & b\cr
                  {\rm 2} & c & \displaystyle d}.
\end{align}

However, behavioural plasticity disturbs game dynamics since both individuals involved in the interaction might behave differently. Hence, the payoff of the strategy 1 (or 2) in a well-mixed population is defined as

\begin{align}
\pi_{1} &= a(x_{1}q_{11} +x_{2}q_{21})+b(x_{1}q_{12}+x_{2}q_{22})\nonumber\\
&= a(x_{1}\gamma +x_{2}(1-\eta))+b(x_{1}(1-\gamma)+x_{2}\eta)\nonumber\\
\pi_{2} &= c(x_{1}q_{11} +x_{2}q_{21})+d(x_{1}q_{12}+x_{2}q_{22}) \nonumber\\
&= c(x_{1}\gamma +x_{2}(1-\eta))+d(x_{1}(1-\gamma)+x_{2}\eta)
\end{align}

\noindent where $x_{1}$ and $x_{2}$ denote the frequencies of A and B types respectively. Since $x_{1} + x_{2} =1$ we use the notation $x = x_{1}$ and $x_{2} = 1-x$ through the rest of the text. Overall payoff values for incompetent types A and B are similarly given by

\begin{align}
\pi_{\rm A} &= q_{11}\pi_{1} + q_{12}\pi_{2}\nonumber\\
&= \gamma \pi_{1} + (1-\gamma)\pi_{2} \nonumber\\
\pi_{\rm B} &= q_{21}\pi_{1} + q_{22}\pi_{2} \nonumber\\
&= (1-\eta) \pi_{1} + \eta\pi_{2}
\end{align}


As a results, an effective payoff matrix can be written as 

\begin{align}
\label{IncompetenceRM}
Q R Q^T &= \bordermatrix{~ & {\rm A} & {\rm B} \cr
                  {\rm A} & \tilde{a} & \tilde{b}\cr
                  {\rm B} & \tilde{c} & \tilde{d}} 
\end{align}

\noindent where

\begin{align}
\tilde{a} &= (a q_{11} + c q_{12})q_{11} + (b q_{11} + d q_{12}) q_{12},\nonumber\\
\tilde{b} &= (a q_{11} + c q_{12})q_{21} + (b q_{11} + d q_{12}) q_{22},\nonumber\\
\tilde{c} &= (a q_{11} + b q_{12})q_{21} + (c q_{11} + d q_{12})q_{22},\nonumber\\
\tilde{d} &= (a q_{21} + b q_{22})q_{21} + (c q_{11} + d q_{12})q_{22}.
\end{align}



\section*{Behavioural plasticity in finite populations}

In the following we discuss the dynamics and how the population of competing plastic types evolve in time. We consider a finite population setting where individuals interact in a stochastic manner. Consider a population consisting of $N$ individuals of 'resident' type $B$ invaded by mutants of type $A$. The total population size is assumed to be constant. At every time step, $t$, the population of mutants, $A$, is given by $n$ individuals, where the population of type $B$ is given by $N-n$ individuals.  

Fitnesses are derived from game payoffs defined in the previous section. Fitness of type $A$ is denoted as $f_{\rm A}$ and type $B$ as $f_{\rm B}$, respectively. The fitnesses are written in terms of payoffs $\pi_{\rm A,B}$ as

\begin{align}
\label{fitnessFP}
    f_{\rm A} = 1+w\pi_{\rm A},\nonumber \\
    f_{\rm B} = 1+w\pi_{\rm B},
\end{align}

\noindent were $w$ is selection intensity. The case of $w=1$ represents strong selection when fitness is completely defined by the payoffs from interactions. For $w\ll 1$, we are in the weak selection limit. This is when the interactions make only a small contribution to the fitness function. 

For the dynamics we utilise a Moran birth-death updating rule on a complete graph \cite{Nowak2006,Nowak2004,Traulsen2008}, which represents a well-mixed population. At every time step, an individual is chosen, proportional to its fitness, to reproduce. The offspring replaces another individual in the population randomly. This model is a Markov process with two absorbing states. The first is an absorbing state of fixation, when a population is taken over by mutants and $n=N$. Similarly, in the case of an extinction absorbing state, the invading mutants cannot successfully compete during the selection process and become extinct ($n=0$). The transition probabilities for such Markov process are given by 

\begin{align}
\label{transprob}
p^{+}_{n} = {\rm Prob}(n \rightarrow +1) &= \frac{f_{\rm A}n}{f_{\rm A}n + f_{\rm B}(N-n)}\frac{N-n}{N}\nonumber\\
p^{-}_{n} = {\rm Prob}(n \rightarrow -1) &= \frac{f_{\rm A}(N-n)}{f_{\rm A}n + f_{\rm B}(N-n)}\frac{n}{N}
\end{align}

The fixation probability of a single randomly placed mutant is thus given by the formula

\begin{align}
\rho_{\rm A} = \frac{1}{\displaystyle 1 + \sum^{N-1}_{k=1}\prod^{k}_{i=1} (p^{-}_{i}/p^{+}_{i})} 
\label{fixation}
\end{align}


In the limit of weak selection we can derive the exact expression for the fixation probability of a plastic mutant $A$ as

\begin{align}
\label{fixprobeq}
    \nicefrac{1}{\rho_A} &\approx N-\frac{1}{6} N w (\gamma +\eta -1) \Big(a (N-2) (\gamma -2 \eta +2)+b (2 \gamma -4 \eta -\gamma  N+2 \eta  N+N+1) \nonumber\\ 
    &+2 \gamma  c-4 c \eta -\gamma  c N+2 c \eta  N-2 c N+c-2 \gamma  d+4 d \eta +\gamma  d N-2 d \eta  N-d N+2 d \Big).
\end{align}

The condition for the selection advantage of plastic mutants is given by

\begin{align}
\rho_A > \frac{1}{N}
\label{fixprobeq2}
\end{align}

\noindent Equation \ref{fixprobeq2} can be further reduced to a simple inequality between the payoff values

\begin{align}
\pi_{\rm A} &> \pi_{\rm B},
\end{align}

\noindent thus, 

\begin{align}
\label{finpopcond}
\nu \pi_{1} > \nu \pi_{2},
\end{align}

\noindent where $\nu:= -1+\gamma+\eta.$ We can re-write this condition further as

\begin{align}
\nu  (a-c) \Big( \gamma n +(1-\eta)(N-n) \Big) > \nu (d-b) \Big( (1-\gamma) n+\eta(N-n)\Big)
\end{align}

\noindent Hence, if $\nu > 0$ then $\pi_{1} > \pi_{2}$ is the condition for selection advantage of a plastic type A. If $\nu < 0$ scenario reverses and $\pi_{2} > \pi_{1}$ is the condition for selection advantage of a plastic type A. Hence, in finite populations, the influence of behavioural uncertainty is determined by the sign of the second eigenvalue of $Q$, which showed to be a global bifurcation in replicator dynamics as well.







One can write down a first-order differential equation for the average frequency of mutants, $x = \langle n\rangle /N$.

\begin{align} 
\dot{x} = \langle p_{+} \rangle - \langle p_{-} \rangle 
\label{averageeq}
\end{align}

\noindent where the brackets $\langle \cdot \rangle$ denote a stochastic average. In the limit of large-$N$ ($N \to \infty)$, equation \ref{averageeq} can be further simplified to,

\begin{align}
\dot{x} = \frac{ f_{A} x (1-x)}{f_{A}x+ f_{B}(1-x)} - \frac{ f_{B} (1-x)x }{f_{A}x+ f_{B}(1-x)}.
\end{align}

This is only exact in infinite populations where the averages over transition probabilities, $\langle p_{+} \rangle \text{ and } \langle p_{-} \rangle$, can be written as $p_{+}(\langle n \rangle)$ and $p_{-}(\langle n \rangle)$. Furthermore, in the limit of weak selection, the denominators in the above equation are approximated with unity and the dynamics is reduced to a replicator equation form as

\begin{align}
\label{RDderived}
\dot{x} = (f_{\rm A} - \phi)x 
\end{align} 

\noindent where $\phi = f_{\rm A}x + f_{\rm B}(1-x)$.

\subsection*{$\nicefrac{1}{3}$ law modification}

The situation becomes more complex in case of coordination games where both strategies are best responses to themselves. Then, the dominating strategy in finite populations is determined by the initial frequency of the mutants. This is often referred to as a $\nicefrac{1}{3}$ law. In a bistable interaction, i.e. such that $a>c$ and $d>b$, there exists an unstable equilibrium where the frequency of type $A$ individuals is given by

\begin{align*}
    x^* = \frac{b-d}{c-a+b-d}.
\end{align*}

Selection prefers type $A$ individuals over type $B$ individuals if the frequency satisfies the inequality $x^*<\nicefrac{1}{3}$. In other words, in a large well-mixed population in the limit of weak selection, the fixation probability of type $A$, $\rho_A$, is greater than $\nicefrac{1}{N}$ if the  frequency of type-A mutants is less than $\nicefrac{1}{3}$. We shall next check how this rule changes under our assumptions. in our settings, the new version of the unstable equilibrium, $\tilde{x}^*$, is given by

\begin{align*}
    \tilde{x}^* = \frac{\eta(b-d)-(1-\eta)(c-a)}{(\gamma - (1-\eta))(c-a+b-d)}
\end{align*}

Thus, the new refined $\nicefrac{1}{3}$ law in the plastic game has the form

\begin{align}
\label{13rule}
   \tilde{x}^* = \frac{x^*(2\eta -1)-(1-\eta)}{\gamma - (1-\eta)} < \frac{1}{3}.
\end{align}

Note that the sign of the eigenvalue of the matrix $Q$, $\nu$, changes the direction of the inequality. For $\nu>0$  the condition can be simplified to 

\begin{align}
\label{LAW}
    x^* < \frac{\gamma+2(1-\eta)}{3(2\eta-1)},
\end{align}

\noindent which is greater or equal to $\nicefrac{1}{3}$ for $\gamma+3 \geq 4\eta$. Note that for $\nu<0$ the inequality sign in \eqref{LAW} is reversed. Then, the region of stability of type A individuals can be larger in comparison to the stability region of the original game, making it easier for mutants to invade. The demonstration of the plastic $\nicefrac{1}{3}$ law for a reward matrix 

\begin{equation*}
R = \begin{blockarray}{ccc}
& A & B \\
\begin{block}{c(cc)}
\; A \; & \;10 \;& \;1 \;\\
\; B \; & \; 9 & 3 \;\\
\end{block}
\end{blockarray}\;
\end{equation*}

\begin{figure}[h!]
\begin{center}
\includegraphics[width=1.0\textwidth]{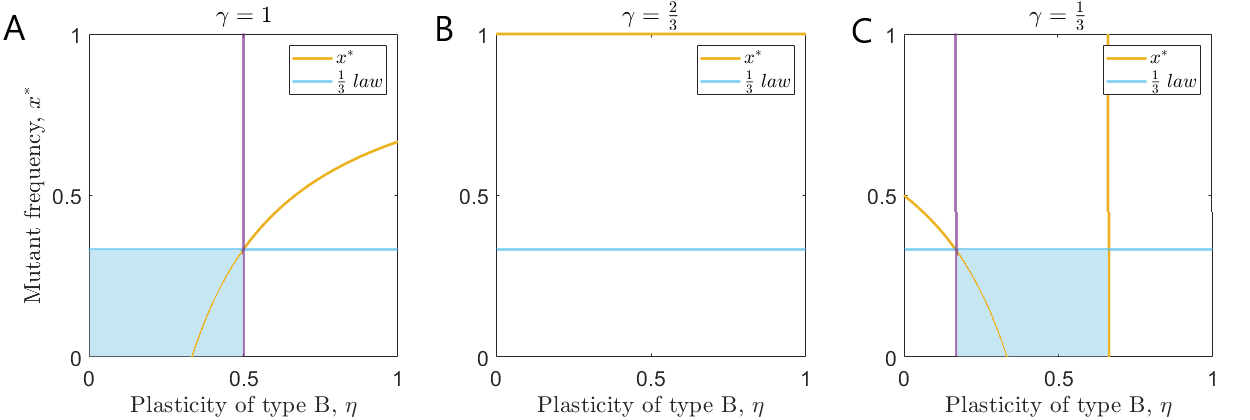}
\end{center}
\caption{Demonstration of $\nicefrac{1}{3}$ law depending on residents' behavioural plasticity. Here, the rule \eqref{13rule} is depicted as a function of $\eta$ for different levels of $\gamma$. (A) $\gamma=1$: type $A$ individuals dominate for the level of $\eta$ less than $\nicefrac{1}{2}$ (shaded area). (B) $\gamma=\nicefrac{2}{3}$: no region of stability is detected for type $A$ individuals. (C) $\gamma=\nicefrac{1}{3}$: there exists a level of behavioural plasticity for type $B$ individuals where type $A$ individuals are preferred by selection.}
\label{fig:1/3rule}
\end{figure}

\noindent can be found in Figure \ref{fig:1/3rule}. Here, $x^*=\nicefrac{2}{3}>\nicefrac{1}{3}$ and hence selection prefers type $B$ individuals over type $A$. However, the selection outcome depends on the mistakes of type $A$ and $B$ individuals. The light-blue region represents the region of stability for type $A$ individuals. As $\lambda$ decreases, the region of stability first shrinks - residents of type $B$ have to be significantly plastic for selection to prefer type $A$. However, after $\gamma=\nicefrac{2}{3}$, the region of stability emerges again and grows as $\gamma\to 0$. These regions exist only for some plasticity level whenever $\tilde{x}^*$ is feasible.

\subsection*{New risk dominance condition}

In addition to the $\nicefrac{1}{3}$ law, one can also check which of the strategies is risk-dominant. In the absence of behavioural plasticity in large populations and weak selection, the condition $\rho_A>\rho_B$ is equivalent to

\begin{equation*}
    a+b>c+d.
\end{equation*}

In a classic case, the risk-dominant strategy has a larger basin of attraction, which can also be written as $x^*<\nicefrac{1}{2}$. In the updated plastic settings we derive the condition for risk dominance of mutants $A$, which is corrected on the actual behaviour of individuals as

\begin{equation*}
    \underbrace{\Big(\gamma+ (1-\eta) \Big)}_{\text{fraction of cooperators}} (a-c) > \underbrace{\Big( (1-\gamma)+\eta \Big)}_{\text{fraction of defectors}} (d-b).
\end{equation*}

In terms of the fixed point for $\nu>0$ this condition can be reduced to

\begin{align}
\label{RDcond}
    x^* < \frac{\gamma+(1-\eta)}{2(2\eta-1)},
\end{align}

In case of $a-c=d-b$, behavioural plasticity would determine the dominant strategy and the more ``frequent'' strategy will overtake. That is, if $\gamma>\eta$ and $a-c=d-b$, then plastic type A is dominating over plastic type B. However, if $a-c\neq d-b$, then degree of plasticity can correct the condition of risk dominance. We summarise plastic concepts of $\nicefrac{1}{3}$ law and risk dominance in Figure \ref{fig:risk}.
However, for general selection and population size, risk dominance does not determine the strategy that will be preferred by selection.

\begin{figure}[h!]
\begin{center}
\includegraphics[width=1.0\textwidth]{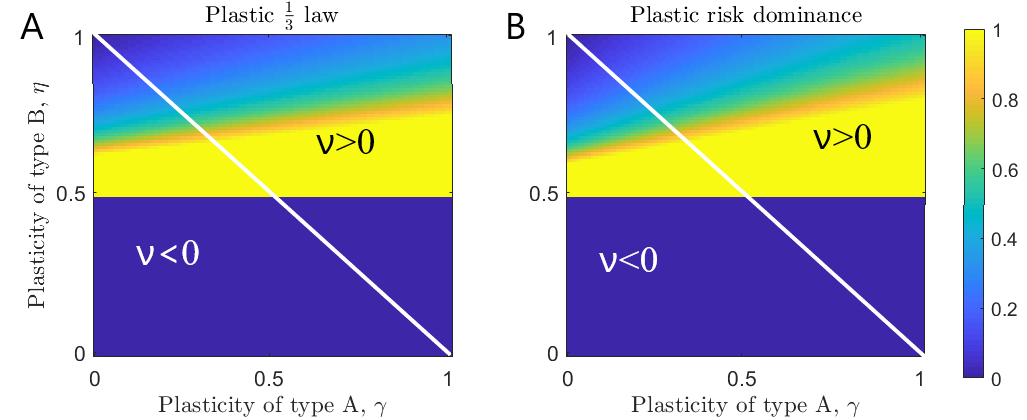}
\end{center}
\caption{Plastic $\nicefrac{1}{3}$ law (panel A) and risk dominance (panel B) conditions. Here, we plot the values of plastic conditions \ref{LAW} and \ref{RDcond} in the plane of types' plasticity. Dark blue region on both plots determines the level of plasticity where conditions are no longer positive, whereas plasticity in bright yellow region provides a chance to type A individuals to overtake type B, which is mostly determined by plasticity of residents. The exact inequality sign depends on the sign of $\nu$: if $\nu>0$, then $x^*$ must be smaller than the value of the conditions. For $\nu<0$ the situation is reversed. The brighter the yellow colour, the higher the value of the conditions for $\nicefrac{1}{3}$ law and risk dominance.}
\label{fig:risk}
\end{figure}


\section*{Behavioural plasticity in infinite populations}
\label{sec:2x2}


In this section, we shall compare our results for finite populations with infinite well-mixed populations. We can write down the replicator equation as derived in \eqref{RDderived} for the frequency of the mutants, $x$ as a function of time $t$ \cite{Taylor1978}. It can be rewritten as

\begin{align}
\dot{x} &= (f_{\rm A} - \phi)x \nonumber\\
&= (f_{\rm A} - f_{\rm B})x(1-x) \nonumber\\
&= w(\pi_{\rm A} - \pi_{\rm B})x(1-x) \nonumber\\
\end{align}

\noindent and simplified further to

\begin{align}
\dot{x} = \{(a-c)[(1-\eta) + (\gamma+\eta -1)x] +(b-d)[\eta -(\gamma+\eta-1)x]\}x(1-x).
\end{align}

\noindent We can read off a so-called $\sigma$-condition for the evolutionary advantage of the mutant type,

\begin{align}
a + \sigma b > c + \sigma d
\end{align}

\noindent where

\begin{align}
\sigma = \frac{ \eta + (\gamma + \eta -1)x^{\star}}{(1-\eta) + (\gamma+\eta-1)x^{\star}}.
\end{align}


\subsection*{Varying behavioural plasticity}

In the case of infinite populations, the time of interaction and size of population are not limited. Hence, plasticity can change over time as species adapt to their environment and specialise. To cover this possibility, we shall assume that individuals are able to improve their strategic execution towards utilising pure strategies. Hence, we let the population to explore their environment via a learning process with a parameter $\lambda$, which represents the probability of using pure strategies. That is we let them adapt from some mixed strategies matrix, $S$, to the perfect execution, an identity matrix $I$, by the law 

\begin{equation}
\label{incomp}
Q(\lambda)=(1-\lambda) S+\lambda I,\; \lambda \in[0,1].
\end{equation}

Hence, the fitness matrix depends on $\lambda$. Generally, we set matrix $S$ to have the form

\begin{equation*}
\label{Smatrix}
S=\left(\begin{array}{cc}
\alpha& 1-\alpha\\
1-\beta& \beta
\end{array} \right).
\end{equation*}

Note that $\gamma$ and $\eta$ from the previous sections can be rewritten in terms of $\alpha$ and $\beta$ as follows: $\gamma=\lambda(1-\alpha)+\alpha$ and $\eta=\lambda(1-\beta)+\beta$. 

Next, let us recall that the structure of the fitness matrix can be reduced by subtracting diagonal elements of the corresponding columns since replicator dynamics remain invariant under such positive linear transformation \cite{Hofbauer1980}. Hence, we can significantly simplify our analysis by considering a matrix with $0$ on the main diagonal, that is,

\begin{equation}
\label{fitRorig}
\tilde{R} = \begin{blockarray}{ccc}
& A & B \\
\begin{block}{c(cc)}
\; A \; & \;0 \;& \;b-d \;\\
\; B \; & \; c-a & 0 \;\\
\end{block}
\end{blockarray}\;.
\end{equation}

Let us next recall the result on all possible game flows in two-dimensional games from \cite{Zeeman1980}. Here, depending on the signs of $c-a$ and $b-d$ in \eqref{fitRorig} we may observe different qualitative behaviour of the game, specifically:
\begin{enumerate}
    \item {\bf If $c>a$ and $d>b$:} there exists no mixed equilibrium and strategy A is dominating.
    \item {\bf If $c<a$ and $b>d$:} there exists no mixed equilibrium and strategy B is dominating.
    \item {\bf If $c>a,b>d$:} there exists a stable mixed equilibrium $(x^*,1-x^*)$ and both vertices are unstable.
    \item {\bf If $a>c,d>b$:} there exists an unstable mixed equilibrium $(x^*,1-x^*)$ and both vertices are stable. The dominating strategy is then determined by the initial frequencies.
    \item {\bf If $c-a=b-d=0$:} the dynamics is trivial and any point $(x,y)$ is stable.
\end{enumerate}

\begin{figure}[h!]
\begin{center}
\includegraphics[width=\textwidth]{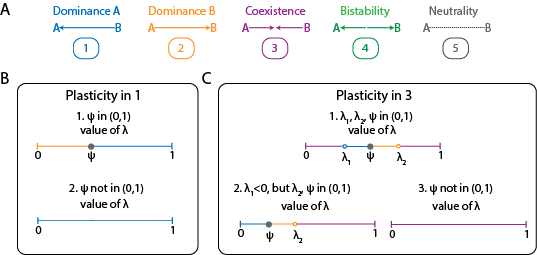}
\end{center}
\caption{Game transitions in 2x2 matrix games in infinite populations. We first demonstrate all possible game flows in two dimensions without plasticity (panel A). Then, we explore the behaviour of our model in the games while the probability of executing pure strategies, $\lambda$, varies in $[0,1]$. We start with dominance of one strategy (panel B), where only one transition is possible at $\lambda^*=\psi$ if $\psi\in[0,1]$. However, in the games with an interior equilibrium we can observe more transitions (panel C). Specifically, if both $\lambda^c_1$ and $\lambda^c_2$ are feasible, then we can observe as many as four game flows: from interior equilibria to pure equilibria to neutrality.}
\label{fig:2x2}
\end{figure}

All these cases are captured in Figure \ref{fig:2x2} panel A. Hence, we may or may not observe a mixed equilibrium $x^*$. In the view of equation \eqref{IncompetenceRM}, for the new incompetent game we obtain a reward matrix that depends on $x^*,\alpha,\beta$ and $\lambda$ and has the following canonical form

\begin{equation}
\label{Rmatrix}
\tilde{R}(\lambda) = (c-a+b-d) \; \nu \times \begin{blockarray}{(cc)}
0 & \beta-(1-x^*)+\lambda(1-\beta)\\
\\
\alpha-x^*+\lambda (1-\alpha) & 0\\
\end{blockarray}\;,
\end{equation}

\noindent where in this case $\nu= \alpha+\beta-1+\lambda(2-\alpha-\beta).$ 
Then, if behaviour of the game is determined by the signs of the elements of the matrix $\tilde{R}(\lambda)$, we shall next analyse all possible transitions that the game passes as plasticity of individuals changes.

Let us first introduce some notation. We can simplify matrix $\tilde{R}(\lambda)$ by introducing two measures combining both fitness and behavioural plasticity effects in the case if $x^*$ is in $(0,1)$ defined as

\begin{align}
    \tilde{a} := \frac{x^*-\alpha}{1-\alpha},\;\;\; \tilde{b} := \frac{1-x^*-\beta}{1-\beta}.
\end{align}

\noindent Note that these values exist only if the game possesses an interior equilibrium $x^*$. In what follows, we compare types $A$ and $B$ based on their advantages, which help us predict which strategy is dominating. 


We say that if $\tilde{a}>\tilde{b}$, then strategy $A$ has a higher advantage, and vice versa. However, there is one more important parameter that has a global influence on the game dynamics: the second eigenvalue of $Q(\lambda)$, $\nu$. We define a new value, $\lambda^*=\psi$ such that $\nu(\psi)=0$, determined as

\begin{equation}
    \psi = \frac{1-\alpha-\beta}{2-\alpha-\beta}.
\end{equation}

\noindent Note that $\psi$ can be written as

\begin{equation}
    \psi = \frac{x^*-\alpha+1-x^*-\beta}{(1-\alpha)+(1-\beta)},
\end{equation}

\noindent which relates to the advantages of both strategies. This value is a mediant of $\tilde{a}$ and $\tilde{b}$, that is,

$$\min(\tilde{a},\tilde{b})<\psi<\max(\tilde{a},\tilde{b}).$$

However, for these values to be able to disrupt the game behaviour we require that $\psi\in[0,1]$, which is possible if and only if $\det(S)<0$. Then, knowing that the critical value of $\lambda$ exists in the interval $[0,1]$, we can show that, in fact, $\lambda^c=\psi$ is a bifurcation point of the game that reflects the stability of the equlibria (see Appendix for more technical details). Hence, value $\psi$ is a global bifurcation of the incompetent game dynamics for any initial game flow.

Further, note that stability of the interior equilibrium depends on the behaviour of the mean fitness given by

\begin{equation*}
    \phi = \frac{(c-a)(b-d)}{c-a+b-d}.
\end{equation*}

\noindent In our new plastic game the mean fitness of the mixed equilibrium 
can be re-scaled to

\begin{equation*}
 \tilde{\phi}(\lambda):= 
 -\tilde{a}\tilde{b} - \tilde{b}\lambda - \tilde{a}\lambda + \lambda^2.
 \end{equation*}
 
The re-scaled mean fitness $\tilde{\phi}(\lambda)$ changes its sign at $\lambda^c=\tilde{b}$ and $\lambda^c=\tilde{a}$. That is, the strategic advantages are, in fact, the bifurcation points of the game (see Appendix for more technical details). Note that these effects are observed only in the games with existing interior equilibrium as these advantages exist only if $x^*,1-x^*\in(0,1)$.

Then, if the original game possesses a stable interior equilibrium and $\tilde{a}>\tilde{b}$, then transitions in a plastic game follow the rules determined below:
\begin{enumerate}
    \item For $\lambda\in[0,\tilde{a})$ there exists a stable mixed equilibrium;
    \item For $\lambda\in[\tilde{a},\psi)$ strategy A dominates strategy B;
    \item At $\lambda=\psi$ the game is neutral;
    \item For $\lambda\in(\psi,\tilde{b}]$ strategy B dominates strategy A;
    \item For $\lambda\in[\tilde{b},1]$ there exists a stable mixed equilibrium again.
\end{enumerate}

If $\tilde{a}<\tilde{b}$, then strategy B dominates A. In the case of an unstable equilibrium, the transitions are equivalent but for $\tilde{a}<\tilde{b}$ first strategy B dominates strategy A. Hence, the advantage determines which strategy will be dominating for small values of the  parameter $\lambda$. 
All transitions for 2x2 matrix games are summarised in Figure \ref{fig:2x2} panels B and C.

\section*{Prisoners' Dilemma}

To demonstrate our model, let us consider an example of Prisoners' Dilemma \cite{Tucker1950,Rapoport1965,Axelrod1981}. Individuals interacting with cooperators obtain benefit $b$ while cooperators pay cost $c$. Then, the payoffs to cooperators and defectors are defined by the following matrix

\begin{align}
R = \bordermatrix{~ & {\rm C} & {\rm D} \cr
                  {\rm C} & \displaystyle b-c & -c\cr
                  {\rm D} & b & \displaystyle 0}.
\end{align}

We introduce the notion of plasticity in the game and consider a new payoff matrix

\begin{align}
\tilde{R} = \left(
\begin{array}{cc}
 (b-c) \gamma  & b(1- \eta) - c \gamma  \\
 b \gamma -c (1-\eta) & (b-c) (1-\eta) \\
\end{array}
\right),
\end{align}
where $\gamma=\alpha(1-\lambda)+\lambda$ and $\eta=\beta(1-\lambda)+\lambda$. The payoffs for cooperators and defectors respectively are given by

\begin{align}
\pi_{\rm C} 
&= b x (\gamma - (1-\eta))+b - b \eta - c \gamma\\
\pi_{\rm D} 
&= b x (\gamma - (1-\eta))+b - b \eta - c (1 -\eta).
\end{align}

The replicator equation in this case has the following form

\begin{align}
\dot{x} &= w(\pi_{\rm C} - \pi_{\rm D})x(1-x) \nonumber\\
\end{align}

\noindent that simplifies to

\begin{align}
\dot{x} = - c (1-x) x (\gamma - (1-\eta)).
\end{align}

Hence, the frequency of cooperators will steadily decline whenever $\gamma > (1-\eta)$. However, if the plasticity level of both populations is such that $\gamma < (1-\eta)$, then cooperators will benefit from plasticity obtaining positive rate of change, which supports our results from the previous section (see Figure \ref{fig:2x2} panel B).

Given \eqref{finpopcond}, we obtain that similar rule is applicable to the finite population invasions. Hence, in the donation game the level of plasticity can help cooperators if defectors are also incompetent. In this case, selection will prefer cooperators over defectors if $\gamma < (1-\eta)$. 
On the other hand, if defectors use their pure strategy, i.e. $\eta=1$, then cooperators are better off by coping defectors' strategy fully and obtaining the highest possible probability of fixation $\rho_C = \nicefrac{1}{N}$.



\begin{figure}[h!]
\begin{center}
\includegraphics[width=1.0\textwidth]{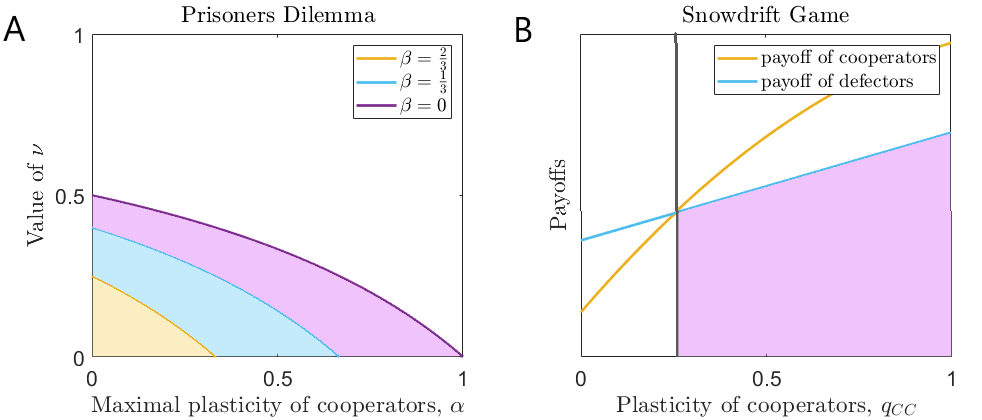}
\end{center}
\caption{Critical levels of $\lambda$ for cooperators survival as a function of defectors' behavioural plasticity in finite populations. (A) The level of critical plasticity required for cooperators survival as a function of defectors' initial plasticity $\beta$ in Prisoners' Dilemma. Here, $\lambda^c=\psi$ is depicted as a function of $\alpha$ and $\beta$. (B) Critical level of cooperators plasticity $q_{CC}$ required for cooperators survival in the Snowdrift game (shaded area). Here, cooperators' and defectors' payoffs are depicted as a function of $\gamma$.}
\label{fig:nu}
\end{figure}

In the case of classic Prisoners' Dilemma in infinite and well-mixed populations, defectors dominate cooperators. Hence, we are in the case of panels B or C in Figure 1. However, once behaviour of players is not certain, for the values of $\lambda<\psi$, the cooperators will denominate defectors. In these settings, dominance of cooperators can be considered as imitation of defectors. With probability $q_{CC}(\lambda)$ cooperators will still cooperate executing their own strategy. 

The value of $\psi$ depends on the matrix $S$ that captures initial degree of behavioural plasticity of individuals. Here, $\alpha$ determines how often type $A$ individuals will play the strategy of their type and $\beta$ - for type $B$ individuals. For example, if matrix $S$ is given by
\begin{equation}
\label{Sexample}
    S = \begin{blockarray}{(cc)}
    \;0.3 & 0.7\;\\
    \\
    \;0.6 & 0.4\;\\
    \end{blockarray}\;,
\end{equation}
which has a negative determinant, then there exists a critical value of $\lambda^c=\psi=\nicefrac{3}{13}$. Hence, the level to which population has to decrease their plasticity ($\lambda^c$) to switch the stability of the equilibria is determined by $\alpha$ and $\beta$. Different values of $\psi$ are depicted in Figure \ref{fig:nu} panel A.



\section*{Snowdrift game}

Next, we consider an example of a Snowdrift game, where cooperation persists at some frequency. This game is also referred to as a Hawk-Dove game or a Chicken game \cite{Smith1973,Smith1982,Nowak2006,Doebeli2005}. Here, the cooperators share their expenses $c$ if both cooperate. Also, even if a cooperator interacts with a defector, the cooperator still obtains benefit $b$. The payoffs to cooperators and defectors are defined by the matrix

\begin{align}
R = \bordermatrix{~ & {\rm C} & {\rm D} \cr
                  {\rm C} & \displaystyle b-\nicefrac{c}{2} & b-c\cr
                  {\rm D} & b & \displaystyle 0}.
\end{align}

Then, the new plastic payoff matrix is given by

\begin{align}
\tilde{R} = \left(
\begin{array}{cc}
 \frac{1}{2} (2 b-c) (2-\gamma) \gamma  & (\gamma -1) \eta  b+b-\frac{1}{2} c \gamma  (\eta +1) \\
 -(1-\gamma) \eta  b+b-\frac{1}{2} c (2-\gamma) (1-\eta) & \frac{1}{2} (2 b-c) \left(1-\eta ^2\right) \\
\end{array}
\right).
\end{align}

Cooperators and defectors in this game obtain the following payoffs

\begin{align*}
\pi_{\rm C}& = b \Big( (\gamma -1) \eta -(\gamma -1) x (\gamma +\eta -1)+1 \Big)+\frac{1}{2} c \gamma  \Big( -1-\eta +x (\gamma +\eta -1) \Big)\\
\pi_{\rm D}& = b (\gamma -1) \eta  x+b \eta ^2 (x-1)+b-\frac{1}{2} c (\eta -1) (-1 -\eta +x (\gamma +\eta -1))
\end{align*}


\noindent and the mixed equilibrium is given by

\begin{equation}
    \tilde{x}^*=\frac{2 b \eta -c (\eta +1)}{(2 b-c) (\gamma - (\eta))}.
\end{equation}

Depending on the mixed strategies of cooperators and defectors, we may observe up to three transitions in the game while players are learning. Let us demonstrate this on an example. For the Snowdrift game, we set a fitness matrix to be

\begin{equation*}
    R = \begin{blockarray}{(cc)}
    \;2 & 1\;\\
    \\
    \;3 & 0\;\\
\end{blockarray}\;,
\end{equation*}

\noindent resulting in a stable mixed equilibrium $x^*=1-x^*=\nicefrac{1}{2}$. We set the starting level of incompetence to $S$ from \eqref{Sexample}. Hence, we are in the case of panel C in Figure 2. The critical values of the incompetence parameter are given by

\begin{equation*}
    \tilde{a} = \nicefrac{2}{7},\quad \psi=\nicefrac{3}{13},\quad\tilde{b}=\nicefrac{1}{6}.
\end{equation*}

The strategic advantage in this case is such that $\tilde{a}<\tilde{b}$. Hence, we obtain five regimes while the probability to play a pure strategy, $\lambda$, changes from 0 to 1. In the first regime, for $\lambda\in[0,\nicefrac{1}{6})$, there exists a stable mixed equilibrium $x^*(\lambda)$. As $\lambda$ grows and reaches the next interval $[\nicefrac{1}{6},\nicefrac{3}{13})$, cooperators dominate and outcompete defectors. The point $\lambda=\nicefrac{3}{13}$ is a transition point where both strategies are neutral, followed by the next critical value $\lambda=\nicefrac{2}{7}$. It switches stability of the equilibria and cooperators get out-competed by defectors. Then, again, for $\lambda\in(\nicefrac{2}{7},1]$, there exists a stable mixed equilibrium $x^*(\lambda)$. The mixed equilibrium here depends on $\lambda$ and is exactly equal to $\tilde{x}^*(\lambda)=1-\tilde{x}^*(\lambda)=\nicefrac{1}{2}$ only for $\lambda=1$.

However, in finite populations the dominating strategy is determined by the condition \eqref{finpopcond}. Hence, in Snowdrift games there will be a critical value of incompetence leading to the switching between defectors and cooperators dominating. We captured this effect in Figure \ref{fig:nu} panel B.


\section*{Comparison with other forms of plasticity}

In evolutionary theory, various forms of strategy plasticity have been discussed with the one most popular being phenotypic plasticity. In general, each strategy is composed of several subtypes that can divide or transform into each other. For simplicity we assume two competing strategies, say strategy A and strategy B, being able to transform or divide into each other with given rates or probabilities. We consider a model where types can divide (a)symmetrically into same offsprings or different ones. For example, type A can divide, with probability $q_{\rm AA}$, into two type A daughter offsprings. With probability $q_{\rm AB}$, it divides to an A daughter and a B daughter cells. Similarly, a type B divides into two type B cells with probability $q_{\rm BB}$ and into asymmetrically into one A and one B with $q_{\rm BA}$. The overall rate that a type $i (i= A,B)$ divides into two daughters $i$ and $j$ types $(i,j = \{A,B\})$, is given by $f_{i} \times q_{ij}$. This is schematically depicted in Figure \ref{fig:schemeSI}.   

\begin{figure}[h!]
\begin{center}
\includegraphics[scale=0.5]{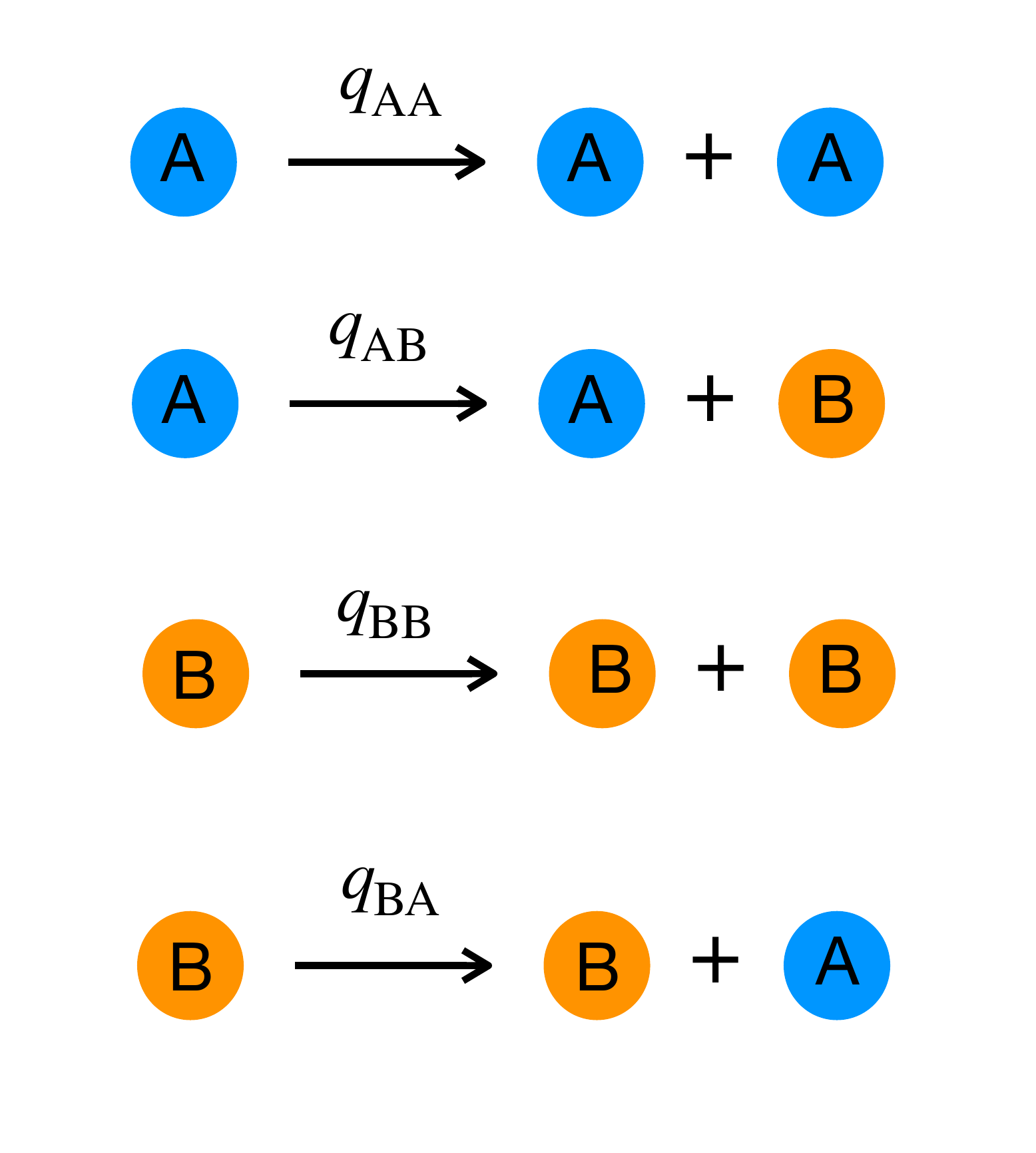}
\end{center}
\caption{ Schematic representation of two plastic types A and B that can divide into each other. Upon a division event, a type A divides with probability $q_{\rm AA}$ into two A daughter cells. With probability $q_{\rm AB}$ it divides asymmetrically into one A and one B daughter. Similar division scheme with probabilities $\eta$ and $q_{\rm BA}$ is true for type B.}
\label{fig:schemeSI}
\end{figure}

As before, we assume that the fitnesses are derived from a matrix game interaction. Recall that fitness of A (B)  is denoted by $f_{\rm A} (f_{\rm B})$ and given in equation \eqref{fitnessFP}. However, the payoffs are now defined differently. Denoting the frequency of type A by $x_{\rm A}$ and type B by $x_{\rm B}$, we have derived new payoffs as


\begin{align}
\pi_{\rm A} &= a \cdot x_{\rm A} + b \cdot x_{\rm B} \nonumber\\
\pi_{\rm B} &= c \cdot x_{\rm A} + d \cdot x_{\rm B}
\end{align}

Then, for finite but large populations the transition probabilities for a plastic model would differ from the one given in equation \eqref{transprob} and be given by 

\begin{align}
p^{+} = {\rm Prob}(n \rightarrow n+1) &= \frac{\big(f_{\rm A}x_{\rm A}\cdot \gamma + f_{\rm B}x_{\rm B}\cdot q_{\rm AB}\big)x_{\rm B}}{f_{\rm A}\cdot x_{\rm A} + f_{\rm B}\cdot x_{\rm B}}\nonumber\\
p^{-} = {\rm Prob}(n \rightarrow n-1) &= \frac{\big(f_{\rm B}x_{\rm B}\cdot \eta + f_{\rm A}x_{\rm A}\cdot q_{\rm AB}\big)x_{\rm A}}{f_{\rm A}\cdot x_{\rm A} + f_{\rm B}\cdot x_{\rm B}} 
\label{transition}
\end{align}

Re-writing the replicator equation using for the rate of change of frequencies $x_{\rm A}$ and $x_{\rm B}$ (recall that $x_{\rm A} + x_{\rm B} = 1$), we obtain

\begin{align}
\dot{x}_{\rm A} &= \langle p^{+} \rangle - \langle p^{-} \rangle \nonumber\\
& =\frac{ f_{\rm A}x_{\rm A} \cdot \gamma + f_{\rm B}x_{\rm B} \cdot q_{\rm BA} - x_{\rm A}\big( f_{\rm A}x_{\rm A} + f_{\rm B}x_{\rm B}\big)}{f_{\rm A} x_{\rm A} + f_{\rm B}x_{\rm B}} 
\end{align}

\noindent where $\langle \cdot \rangle$ denotes the average over stochastic variable. We assume large populations so the average over powers of $x_{\rm A,B}$ is replaced with their average to that power. We could derive the same replicator equation up to the numerator in infinite populations using standard methods as 

\begin{align}
\dot{x} &= f_{\rm A}x \cdot \gamma + f_{\rm B}(1-x)\cdot q_{\rm BA} - \phi x \nonumber\\
&= f_{\rm A}x (\gamma - x) + f_{\rm B}(1-x)(q_{\rm BA} - x) -\phi x
\label{replicatorSI}
\end{align} 

\noindent where $x = x_{\rm A}$ and $1-x = x_{\rm B}$ and $\phi$ is the mean fitness defined as before $\phi = f_{\rm A}x + f_{\rm B}(1-x)$. This would corresponds to a weak selection limit for the Moran process where the $f_{\rm A}x_{\rm A} + f_{\rm B}x_{\rm B} \approx 1$. We refer the reader to Appendix for more technical details of the model and its analysis.

We show that nothing similar to generalised version of \nicefrac{1}{3} law can be observed in this model of phenotypic plasticity indicating the difference from our model of incompetence. This is partly due to the fact different that subtypes produce offsprings of other subtypes and the phenotypic plasticity does manifest itself during the game interaction/encounter. Similar conditions can be derived for the ESS condition in this case but for the sake of brevity and avoiding deviations from the notion of incompetence we stop at this level of analysis here.



\section*{Conclusions}

We considered 2-dimensional evolutionary games where players may execute unintended actions when interacting with their opponents, which we call behavioural plasticity or incompetence. We derived all possible game transitions in such games in completely mixed infinite populations' dynamics. We showed that behavioural mistakes of individuals may change the outcome of the game and disturb stability of equilibria. This can happen through a process, where individuals are allowed to improve their probability to execute their pure strategy, reducing behavioural mistakes resulting in unintended actions. For example, if the original game possesses only pure strategy equilibrium, then stability regimes may switch leading to the opponents' strategy dominating. However, in the case of mixed equilibrium games, behavioural plasticity can lead to either one or another strategy dominating or a mixed equilibrium existing. The actual outcome is then determined by the exact degree of mistakes and the mixed strategies adopted by individuals.

However, in more natural settings of finite populations' dynamics, the fitness values are derived from game interactions. We assumed that type $A$ individuals invade a resident type $B$ population to answer the question of what should the degree of behavioural plasticity the invading strategy be to maximise their  fixation probability. We show that such an assumption can help to promote cooperation. Though plasticity itself plays an important role in evolutionary modelling, the possibility of cooperators to imitate by mistake a defecting strategy has not yet been considered. Here, mistakes during the invasion may represent a defensive mechanism helping invaders to blend in the resident population. Once established, the invaders re-learn their own strategy becoming more cooperative after the competition pressure is neglected. Specifically, plasticity may ease the $\nicefrac{1}{3}$ law for invaders by stretching this interval. 

Next, we consider examples of Prisoner's Dilemma and a Snowdrift game. For Prisoners' Dilemma we obtain conditions that provide cooperators a selective advantage. Those conditions depend on levels of behavioural plasticity of both cooperators and defectors. Interestingly, if we assume that both defectors and cooperators are prone to behavioural mistakes, then defectors can successfully invade if their probability to execute a pure cooperating strategy is greater than the probability of defectors acting as cooperators.
Further, we apply a similar analysis to the Snowdrift game and derive conditions on when cooperators can successfully invade defectors. For finite populations, there exists a critical mass of behavioural uncertainty of both cooperators and defectors that would allow defectors resist invasion. However, once this critical level is passed, cooperation can be sustained.

In the concluding section, we compared the notion of incompetence (or behavioural plasticity) with phenotypic plasticity. We derived a plastic model where individuals can give birth to individuals of other types. We show that such a model does not imply the existence of $\nicefrac{1}{3}$ law as in the case with incompetence. 

\section*{Funding}

This work was supported by the European Union's Horizon 2020 research and innovation program under the Marie Sklodowska-Curie Grant Agreement \#754411.

{
\bibliography{MyRefs} 

\begin{thebibliography}{10}

\bibitem{Smith1973}
J.~Smith and G.~Price, ``The logic of animal conflict,'' {\em Nature},
  vol.~246, pp.~15--18, 1973.

\bibitem{Beck2013}
J.~Beck, {\em Incompetence, training and changing capabilities in {Game}
  theory}.
\newblock PhD thesis, University of South Australia, Australia, 2013.

\bibitem{Filar2012}
J.~D. Beck, V.~Ejov, and J.~A. Filar, ``Incompetence and impact of training in
  bimatrix games,'' {\em Automatica}, vol.~48, no.~10, pp.~2400--2408, 2012.

\bibitem{Selten1975}
R.~Selten, ``Reexamination of the perfectness concept for equilibrium points in
  extensive games,'' {\em International Journal of Game Theory}, vol.~4, no.~1,
  pp.~25--55, 1975.

\bibitem{Stadler1992}
P.~Stadler and P.~Schuster, ``Mutation in autocatalytic reaction networks,''
  {\em Journal of mathematical biology}, vol.~30, no.~6, pp.~597--631, 1992.

\bibitem{Tarnita2009}
C.~E. Tarnita, T.~Antal, and M.~A. Nowak, ``Mutation--selection equilibrium in
  games with mixed strategies,'' {\em Journal of theoretical biology},
  vol.~261, no.~1, pp.~50--57, 2009.

\bibitem{Komarova2004}
N.~Komarova, ``Replicator-mutator equation, universality property and
  population dynamics of learning,'' {\em Journal of theoretical biology},
  vol.~230, pp.~227--239, 2004.

\bibitem{Komarova2001}
N.~Komarova, P.Niyogi, and M.~Nowak, ``The evolutionary dynamics of grammar
  acquisition,'' {\em Journal of theoretical biology}, vol.~209, pp.~43--59,
  2001.

\bibitem{Nowak2001}
M.~Nowak, N.~Komarova, and P.~Niyogi, ``Evolution of universal grammar,'' {\em
  Science}, vol.~291, no.~5501, pp.~114--118, 2001.

\bibitem{FudenbergLevine}
D.~Fudenberg and D.~Levine, {\em The theory of learning in games}.
\newblock USA: The MIT Press, 1999.

\bibitem{Hopkins2002}
E.~Hopkins, ``Two competing models of how people learn in games,'' {\em
  Econometrica}, vol.~70, no.~6, pp.~2141--2166, 2002.

\bibitem{McKelvey1995}
R.~McKelvey and T.~Palfrey, ``Quantal response equilibria for normal form
  games,'' {\em Games and Economic Behavior}, vol.~10, pp.~6--38, 1995.

\bibitem{Selten1991}
R.~Selten, ``Evolution, {Learning}, and {Economic} {Behavior},'' {\em Games and
  Economic Behavior}, vol.~3, pp.~3--24, 1991.

\bibitem{Levin2003}
S.~Levin, ``Complex adaptive systems: Exploring the known, the unknown and the
  unknowable,'' {\em Bulletin of the American Mathematical Society}, vol.~40,
  no.~1, pp.~3--19, 2003.

\bibitem{Dridi2019}
S.~Dridi, ``Plasticity in evolutionary games,'' {\em bioRxiv}, p.~509604, 2019.

\bibitem{Su2016}
Q.~Su, A.~Li, L.~Zhou, and L.~Wang, ``Interactive diversity promotes the
  evolution of cooperation in structured populations,'' {\em New Journal of
  Physics}, vol.~18, no.~10, p.~103007, 2016.

\bibitem{Su2019}
Q.~Su, L.~Zhou, and L.~Wang, ``Evolutionary multiplayer games on graphs with
  edge diversity,'' {\em PLoS computational biology}, vol.~15, no.~4,
  p.~e1006947, 2019.

\bibitem{Bomze1995}
I.~Bomze and R.~Burger, ``Stability by mutation in evolutionary games,'' {\em
  Games and Economic Behavior}, vol.~11, no.~2, pp.~146--172, 1995.

\bibitem{Apaloo2009}
J.~Apaloo, J.~S. Brown, and T.~L. Vincent, ``Evolutionary game theory: {ESS},
  convergence stability, and {NIS},'' {\em Evolutionary Ecology Research},
  vol.~11, no.~4, pp.~489--515, 2009.

\bibitem{Hofbauer2003}
J.~Hofbauer and K.~Sigmund, ``Evolutionary game dynamics,'' {\em Bulletin of
  the American Mathematical Society}, vol.~40, no.~3, pp.~479--519, 2003.

\bibitem{Apaloo1995}
R.~McKelvey and J.~Apaloo, ``The structure and evolution of
  competition-organized ecological communities,'' {\em The Rocky Mountain
  Journal of Mathematics}, vol.~25, no.~1, pp.~417--436, 1995.

\bibitem{Nowak2006}
M.~Nowak, {\em Evolutionary dynamics: exploring the equations of life}.
\newblock UK: The Belknap press of Harvard University press, 2006.

\bibitem{Taylor1978}
P.~Taylor and L.~Jonker, ``Evolutionary stable strategies and {Game}
  {Dynamics},'' {\em Mathematical Biosciences}, vol.~40, pp.~145--156, 1978.

\bibitem{Zeeman1980}
E.~Zeeman, ``Population dynamics from game theory,'' in {\em Global Theory of
  Dynamical Systems}, pp.~471--497, Springer, 1980.

\bibitem{Smith1982}
J.~Smith, {\em Evolution and the Theory of Games}.
\newblock USA: Cambridge University Press, 1982.

\bibitem{Durrett1994}
R.~Durrett and S.~Levin, ``The importance of being discrete (and spatial),''
  {\em Theoretical population biology}, vol.~46, no.~3, pp.~363--394, 1994.

\bibitem{Nowak2004a}
M.~A. Nowak and K.~Sigmund, ``Evolutionary dynamics of biological games,'' {\em
  science}, vol.~303, no.~5659, pp.~793--799, 2004.

\bibitem{Allen2017}
B.~Allen, G.~Lippner, Y.-T. Chen, B.~Fotouhi, N.~Momeni, S.-T. Yau, and M.~A.
  Nowak, ``Evolutionary dynamics on any population structure,'' {\em Nature},
  vol.~544, no.~7649, pp.~227--230, 2017.

\bibitem{Kaveh2019}
K.~Kaveh, A.~McAvoy, and M.~A. Nowak, ``Environmental fitness heterogeneity in
  the moran process,'' {\em Royal Society open science}, vol.~6, no.~1,
  p.~181661, 2019.

\bibitem{Ohtsuki2007}
H.~Ohtsuki, P.~Bordalo, and M.~A. Nowak, ``The one-third law of evolutionary
  dynamics,'' {\em Journal of theoretical biology}, vol.~249, no.~2,
  pp.~289--295, 2007.

\bibitem{carja2017evolutionary}
O.~Carja and J.~B. Plotkin, ``The evolutionary advantage of heritable
  phenotypic heterogeneity,'' {\em Scientific reports}, vol.~7, no.~1,
  pp.~1--12, 2017.

\bibitem{raj2008nature}
A.~Raj and A.~van Oudenaarden, ``Nature, nurture, or chance: stochastic gene
  expression and its consequences,'' {\em Cell}, vol.~135, no.~2, pp.~216--226,
  2008.

\bibitem{acar2005enhancement}
M.~Acar, A.~Becskei, and A.~van Oudenaarden, ``Enhancement of cellular memory
  by reducing stochastic transitions,'' {\em Nature}, vol.~435, no.~7039,
  pp.~228--232, 2005.

\bibitem{kussell2005phenotypic}
E.~Kussell and S.~Leibler, ``Phenotypic diversity, population growth, and
  information in fluctuating environments,'' {\em Science}, vol.~309, no.~5743,
  pp.~2075--2078, 2005.

\bibitem{gallie2015bistability}
J.~Gallie, E.~Libby, F.~Bertels, P.~Remigi, C.~B. Jendresen, G.~C. Ferguson,
  N.~Desprat, M.~F. Buffing, U.~Sauer, H.~J. Beaumont, {\em et~al.},
  ``Bistability in a metabolic network underpins the de novo evolution of
  colony switching in pseudomonas fluorescens,'' {\em PLoS Biol}, vol.~13,
  no.~3, p.~e1002109, 2015.

\bibitem{cohen2013microbial}
N.~R. Cohen, M.~A. Lobritz, and J.~J. Collins, ``Microbial persistence and the
  road to drug resistance,'' {\em Cell host \& microbe}, vol.~13, no.~6,
  pp.~632--642, 2013.

\bibitem{kaveh2017stem}
K.~Kaveh, ``Stem cell evolutionary dynamics of differentiation and
  plasticity,'' {\em Current Stem Cell Reports}, vol.~3, no.~4, pp.~366--372,
  2017.

\bibitem{Doebeli2005}
M.~Doebeli and C.~Hauert, ``Models of cooperation based on the prisoner's
  dilemma and the snowdrift game,'' {\em Ecology letters}, vol.~8, no.~7,
  pp.~748--766, 2005.

\bibitem{Nowak2005}
M.~A. Nowak and K.~Sigmund, ``Evolution of indirect reciprocity,'' {\em
  Nature}, vol.~437, no.~7063, pp.~1291--1298, 2005.

\bibitem{Sigmund2010selfish}
K.~Sigmund, {\em The calculus of selfishness}, vol.~6.
\newblock Princeton University Press, 2010.

\bibitem{Sigmund2012}
K.~Sigmund, ``Moral assessment in indirect reciprocity,'' {\em Journal of
  theoretical biology}, vol.~299, pp.~25--30, 2012.

\bibitem{Hilbe2018}
C.~Hilbe, L.~Schmid, J.~Tkadlec, K.~Chatterjee, and M.~A. Nowak, ``Indirect
  reciprocity with private, noisy, and incomplete information,'' {\em
  Proceedings of the National Academy of Sciences}, vol.~115, no.~48,
  pp.~12241--12246, 2018.

\bibitem{Kleshnina2017}
M.~Kleshnina, J.~A. Filar, V.~Ejov, and J.~C. McKerral, ``Evolutionary games
  under incompetence,'' {\em Journal of mathematical biology}, vol.~77, no.~3,
  pp.~627--646, 2018.

\bibitem{Nowak2004}
M.~A. Nowak, A.~Sasaki, C.~Taylor, and D.~Fudenberg, ``Emergence of cooperation
  and evolutionary stability in finite populations,'' {\em Nature}, vol.~428,
  no.~6983, pp.~646--650, 2004.

\bibitem{Traulsen2008}
A.~Traulsen, N.~Shoresh, and M.~A. Nowak, ``Analytical results for individual
  and group selection of any intensity,'' {\em Bulletin of mathematical
  biology}, vol.~70, no.~5, p.~1410, 2008.

\bibitem{Hofbauer1980}
J.~Hofbauer, P.~Schuster, K.~Sigmund, and R.~Wolff, ``Dynamical systems under
  constant organization ii: Homogeneous growth functions of degree p=2,'' {\em
  SIAM Journal on Applied Mathematics}, vol.~38, no.~2, pp.~282--304, 1980.

\bibitem{Tucker1950}
A.~Tucker, ``A two person dilemma. lecture at stanford university,'' {\em
  Prisoner’s Dilemma, 2nd Edition. Anchor Books, New York}, 1950.

\bibitem{Rapoport1965}
A.~Rapoport, A.~M. Chammah, and C.~J. Orwant, {\em Prisoner's dilemma: A study
  in conflict and cooperation}, vol.~165.
\newblock University of Michigan press, 1965.

\bibitem{Axelrod1981}
R.~Axelrod, ``The emergence of cooperation among egoists,'' {\em American
  political science review}, vol.~75, no.~2, pp.~306--318, 1981.

\end{thebibliography}
\bibliographystyle{ieeetr}
}

\section*{Appendix. Derivation of the results in the main text}

\subsection*{Infinite populations}

\begin{proposition}
\label{res:det}
The critical value $\lambda^*=\psi$ is in the interval $[0,1]$ if and only if $\det(S)<0.$ Moreover, it is a singular point of the matrix of incompetence $Q(\lambda)$.
\end{proposition}

\begin{proof}

Let us first consider the determinant of the starting level of incompetence:

$$\det(S) = \alpha\beta - (1-\beta)(1-\alpha) = -1+\alpha+\beta.$$

\noindent Let us now consider $\psi$:

$$\psi = \frac{1-\alpha-\beta}{2-\alpha-\beta} = \frac{-\det(S)}{2-\alpha-\beta}.$$
    
Next, let us consider the matrix of incompetence, $Q$. The eigenvalues of this matrix are $1$ and $\nu=-1+\alpha+\beta+\lambda(2-\alpha-\beta)$. Hence, the determinant of $Q$ changes its sign from negative to positive at $\lambda^* = \psi$.

\end{proof}

\begin{proposition}
\label{res:signR}
Let, for $\epsilon>0:\min(\tilde{a},\tilde{b})<\psi\pm\epsilon<\max(\tilde{a},\tilde{b})$, the values of the incompetence parameter be defined as $\lambda^-:=\psi-\epsilon$ and $\lambda^+:=\psi+\epsilon$. Then, for such values, the sign of the elements in the fitness matrix $\tilde{R}(\lambda)$ are reversed, that is, $\text{sign}(\tilde{r}_{ij}(\lambda^-))=-\text{sign}(\tilde{r}_{ij}(\lambda^+))$.
\end{proposition}

\begin{proof}

Let us rewrite the fitness matrix from \eqref{Rmatrix} in the following manner:

\begin{equation*}
\label{Rmatrix}
\tilde{R}(\lambda) = (c-a+b-d) \; \nu \times \begin{blockarray}{(cc)}
0 & \hat{r}_{12}\\
\\
\hat{r}_{21} & 0\\
\end{blockarray}\;.
\end{equation*}

    First, let us consider elements of $\tilde{R}(\lambda)$ at $\lambda^*=\psi$:
    
\begin{align*}
    \hat{r}_{12}(\psi) &= \beta-x_{\rm R}+\psi(1-\beta) = \frac{(1-\alpha)(\beta-x_{\rm R})-(1-\beta)(\alpha-x_{\rm M})}{2-\alpha-\beta}:=r,\\
    \hat{r}_{21}(\psi) &= \alpha-x_{\rm M}+\psi(1-\alpha) = \frac{(1-\beta)(\alpha-x_{\rm M})-(1-\alpha)(\beta-x_{\rm R})}{2-\alpha-\beta}=-r.
\end{align*}

Hence, both elements have opposite signs and those are determined by the sign of $\delta=\tilde{a}-\tilde{b}$. Specifically, $r>0  \iff \delta>0$. Note that $\hat{r}_{12}(\psi \pm \epsilon) = r \pm \epsilon(1-\beta)$ and $\hat{r}_{21}(\psi \pm \epsilon) = - r \pm \epsilon(1-\alpha)$. Hence, for $\epsilon$ such that $\min(\tilde{a},\tilde{b})<\psi\pm\epsilon<\max(\tilde{a},\tilde{b})$, the signs of $\hat{r}_{12}(\psi \pm \epsilon)$ and $\hat{r}_{21}(\psi \pm \epsilon)$ are preserved. Then, we obtain

\begin{equation*}
\label{Rmatrix}
\tilde{R}(\psi \pm \epsilon) = (c-a+b-d) \; \nu(\psi\pm\epsilon) \times \begin{blockarray}{(cc)}
0 & \hat{r}_{12}(\psi\pm\epsilon)\\
\\
\hat{r}_{21}(\psi\pm\epsilon) & 0\\
\end{blockarray}\;,
\end{equation*}

Further, the global sign of $\tilde{R}(\lambda)$ is determined by the initial signs of $a$ and $b$ and it does not depend on $\lambda$. In addition, note that $\nu(\psi-\epsilon)<0$ and $\nu(\psi+\epsilon)>0$. Hence, the transition in the signs of the elements in the fitness $\tilde{R}(\lambda)$ matrix happen at $\lambda^*=\psi$.
    
\end{proof}

\begin{proposition}
\label{res:bif}
Let $\lambda^c_1=\tilde{a}$ and $\lambda^c_2=\tilde{b}$ be bifurcation points of the dynamics for $\tilde{R}(\lambda)$ such that $\lambda^c_1,\lambda^c_2\in [0,1]$. Then, transitions of the equlibria are as follows:
\begin{enumerate}
    \item[(i)] If there exists a stable interior equilibrium, then at $\min(\lambda^c_1,\lambda^c_2)$ the game transits from a stable interior equilibrium to a stable pure strategy. If $\delta<0$, then the pure stable strategy is strategy 1, otherwise, strategy 2.
    \item[(ii)] If there exists a stable interior equilibrium, then at $\max(\lambda^c_1,\lambda^c_2)$ the game transits from a stable pure strategy equilibrium to a stable interior equilibrium. If $\delta>0$, then the pure strategy equilibrium is strategy 1, otherwise, strategy 2.
    \item[(iii)] If there exists an unstable interior equilibrium, then at $\min(\lambda^c_1,\lambda^c_2)$ the game transits from an unstable interior equilibrium to a stable pure strategy. If $\delta<0$, then the pure stable strategy is strategy 2, otherwise, strategy 1.
    \item[(iv)] If there exists a stable interior equilibrium, then at $\max(\lambda^c_1,\lambda^c_2)$ the game transits from a stable pure strategy equilibrium to an unstable interior equilibrium. If $\delta>0$, then the pure strategy equilibrium is strategy 2, otherwise, strategy 1.
\end{enumerate}
\end{proposition}

\begin{proof}
    
(i) For a case of a stable interior equilibrium we have $A,B>0$. Without loss of generality, let us assume that $\min(\tilde{a},\tilde{b})=\tilde{a}$, i.e. $\delta<0$. Let us first consider a bifurcation point $\lambda^c_1=\tilde{a}$. The eigenvalue of matrix $Q(\lambda^c_1), \nu(\lambda^c_1)$, is negative since $\tilde{a}<\psi$. The fitness matrix then has the following form

\begin{equation*}
\label{Rmatrix}
    \tilde{R}(\lambda^c_1) = (c-a+b-d) \; \nu(\lambda^c_1) \times \begin{blockarray}{(cc)}
    0 & \bar{r}_{12}(\lambda^c_1)\\
    \\
    \bar{r}_{21}(\lambda^c_1) & 0\\
\end{blockarray}\;,
\end{equation*}

\noindent where $\bar{r}_{12}(\lambda^c_1) = \delta(1-\beta)$ and $\bar{r}_{21}(\lambda^c_1)=0$. Since $\tilde{a}<\tilde{b}$, we have $\bar{r}_{12}(\lambda^c_1)<0$. Next, consider

\begin{align*}
\bar{r}_{21}(\lambda^c_1-\epsilon)=-\epsilon(1-\alpha),\\
\bar{r}_{21}(\lambda^c_1+\epsilon)=\epsilon(1-\alpha).
\end{align*}

Hence, if $\min(\tilde{a},\tilde{b})=\tilde{a}$, then the system bifurcates from a stable interior equilibrium to a stable pure strategy 1.
    
(ii) Let us now assume that $\delta>0$, i.e. $\max(\tilde{a},\tilde{b})=\tilde{b}$. Since $\tilde{b}>\psi$, we obtain a positive eigenvalue $\nu(\lambda^c_2)$. The fitness matrix then is defined as

\begin{equation*}
\label{Rmatrix}
    \tilde{R}(\lambda^c_2) = (c-a+b-d) \; \nu(\lambda^c_2) \times \begin{blockarray}{(cc)}
    0 & \bar{r}_{12}(\lambda^c_2)\\
    \\
    \bar{r}_{21}(\lambda^c_2) & 0\\
\end{blockarray}\;,
\end{equation*}

\noindent where $\bar{r}_{12}(\lambda^c_2)=0$ and $\bar{r}_{12}(\lambda^c_1) = -\delta(1-\alpha)$, which implies $\bar{r}_{12}>0$ and strategy 2 dominates for $\lambda^c_2-\epsilon$. Hence, a system bifurcates from a stable pure strategy 2 to an interior equilibrium.
    
(iii)-(iv) These two parts follow the same argument with the difference that $A,B<0$. Hence, the sign of the matrix $\tilde{R}(\lambda)$ is reversed.
    
\end{proof}

\begin{theorem}
\label{res:thm1}
If $x^*\in(0,1)$ and all the advantages are such that $\psi,\tilde{a},\tilde{b}\in(0,1)$ and $\delta<0$ ($\delta>0$), then the bifurcations happen in the following order: at $\lambda^c_1=\tilde{a}$, $\lambda^c_2=\psi$ and then $\lambda^c_3=\tilde{b}$ ($\lambda^c_1=\tilde{b}$, $\lambda^c_2=\psi$ and then $\lambda^c_3=\tilde{a}$, respectively).
\end{theorem}
    
\begin{proof}
The statement follows immediately from Propositions 1-3.
\end{proof}

\begin{theorem}
\label{res:thm2}
If an interior equilibrium does not exist in the original game (either $a<0,b>0$ or $a>0,b<0$) and $\det(S)<0$, then the only bifurcation value $\lambda^c=\psi$ will switch the stability of the pure equilibria.
\end{theorem}
    
\begin{proof}
If $\det(S)<0$ then $\psi\in[0,1]$ by Proposition 1. Further, if either $a<0,b>0$ or $a>0,b<0$ then $x^*,1-x^*$ do not exist in $(0,1)$ by Lemma 1. Hence, the statement follows from Proposition 2.
\end{proof}

\subsection*{Phenotypic plasticity}

We could derive the same replicator equation up to the numerator in infinite populations using standard methods as well. (This would corresponds to a weak selection limit for the Moran process where the $f_{\rm A}x_{\rm A} + f_{\rm B}x_{\rm B} \approx 1$.) 

\begin{align}
\dot{x} &= f_{\rm A}x \cdot q_{\rm AA} + f_{\rm B}(1-x)\cdot q_{\rm BA} - \phi x \nonumber\\
&= f_{\rm A}x (q_{\rm AA} - x) + f_{\rm B}(1-x)(q_{\rm BA} - x) -\phi x
\label{replicatorSI}
\end{align} 

\noindent where $x = x_{\rm A}$ and $1-x = x_{\rm B}$ and $\phi$ is the mean fitness defined as, $\phi = f_{\rm A}x + f_{\rm B}(1-x)$.

We can define parameters $\alpha$ and $\beta$ similar to the way we defined matrix $S$ in the main text:

\begin{equation*}
\label{qmatrix}
Q=\left(\begin{array}{cc}
\alpha& 1-\alpha\\
1-\beta& \beta
\end{array} \right).
\end{equation*}

Thus $q_{\rm AA} = \alpha, q_{\rm BA} = 1- \alpha, q_{\rm AB} = 1-\beta, \eta = \beta$. The replicator equation in the presence of plasticity can be further simplified to,

\begin{align}
\dot{x} &= (f_{\rm A} - f_{\rm B}) x(1-x) + (f_{\rm B}q_{\rm BA} + f_{\rm A})(q_{\rm AA}-1) x \nonumber\\
&= (f_{\rm A} - f_{\rm B}) x(1-x) + f_{\rm B} (1-\beta) - f_{\rm A} \alpha
\end{align}

Even in the absence of the game interactions the fixed point is changed from boundary points $x^{\star} =0,1$. For $w=0$ it is $x^{\star} = q_{\rm BA}/(1-q_{\rm AA}+q_{\rm BA}) = (1-\beta)/(2- \beta - \alpha)$. For weak selection we can find the fixed point by expanding, $x^{\star}$ in first order of $w$: $x^{\star} = x^{\star}_{0} + w \cdot x^{\star}_{1}$. Where $x^{\star}_{0} = (1-\beta)/(2-\beta - \alpha)$. 

\begin{align}
x^{\star}_{1} = \frac{((a-c)(1-\beta) + (b-d)(1-\alpha))(1-\alpha)(1-\beta)}{(2-\alpha-\beta)^{3}}
\end{align}

The above solutions represent the shift in the boundary fixed points (extinction and fixation) as we introduce plasticity. To find the changes in the interior equilibrium fixed point, equation \ref{13rule}, we expand the solutions for small $1-\alpha$ and $1-\beta$ around the interior equilibrium non-plastic solution, $x^{\star}_{0} = x^{\star}(\alpha=1,\beta=1) =  (d-b)/(a-b-c+d)$. We write the solutions as $x^{\star} = x^{\star}_{0} + \epsilon * x^{\star}_{1}$ where $\epsilon$ is defined as $1- \alpha = \epsilon (1-\hat{\alpha}), 1-\beta = \epsilon (1- \hat{\beta})$. The results are,

\begin{align}
x^{\star} = \frac{d-b}{a-b-c+d} + \frac{(a d w - b c w + a - b - c + d)(a (1-\beta) + (1-\alpha)b - (1-\alpha)d - (1-\beta)c)}{(a - b - c + d)w(b - d)(a - c)}
 \end{align} 

The second term are the correction to the interior fixed point solutions for evolutionary games that follows the 1/3-rules, equation \ref{13rule}.

\end{document}